\documentclass[10pt,conference]{IEEEtran}

\usepackage{amsthm,amsmath,amsfonts,braket,algorithm,graphicx, braket,balance}

\theoremstyle{definition} 
\newtheorem {theorem} {Theorem}

\newtheorem {lemma} {Lemma}

\usepackage{ifthen}
\newboolean{fullversionflag}
\setboolean{fullversionflag}{true}
\newcommand{\fullversion}[2]{\ifthenelse{\boolean{fullversionflag}}{{#1}}{{#2}}}

\newcommand{\heading}[1]{\text{ }\newline\textbf{{#1:}}}

\newcommand{\num}{\#}

\newcommand{\kb}[1]{\left[#1\right]}

\newcommand{\trd}[1]{\left|\left| #1 \right| \right|}

\newcommand{\st}{\text{ } : \text{ }}

\newcommand{\Hmin}{H_\infty}

\newcommand{\leakEC}{\lambda_{EC}}

\newcommand{\up}[1]{^{({#1})}}

\newcommand{\wt}[1]{\widetilde{{#1}}}

\floatname{algorithm}{Protocol}

\newcommand{\filter}{\mathcal{F}}

\newcommand{\samp}{\Psi}

\setboolean{fullversionflag}{true}

\def\BibTeX{{\rm B\kern-.05em{\sc i\kern-.025em b}\kern-.08em
    T\kern-.1667em\lower.7ex\hbox{E}\kern-.125emX}}
\begin{document}

\title{Finite Key Security of the Extended B92 Protocol}

\author{\IEEEauthorblockN{Walter O. Krawec}
\IEEEauthorblockA{\textit{School of Computing} \\
\textit{University of Connecticut}\\
Storrs CT, USA \\
walter.krawec@uconn.edu}}

\maketitle

\begin{abstract}
In this paper, we derive a new proof of security for the Extended B92 QKD protocol.  We derive a general entropic uncertainty relation for QKD protocols with data filtering and rejection.  Our bound requires one to determine the size of a particular set derived from a classical sampling strategy.  Finally, we show how our methods can be used to readily prove security of the Extended B92 protocol, providing, to our knowledge, the first finite key proof of security for this protocol against general, coherent, attacks.
\end{abstract}


\section{Introduction}

Quantum cryptography is a fascinating field, and holds several potential practical advantages, along with numerous, interesting, theoretical problems.  Unlike classical cryptography (including post-quantum cryptography) which typically, or necessarily in many cases, relies on computational assumptions for security, quantum key distribution (QKD), and other quantum cryptographic protocols, can be proven information theoretic secure.  Typically, the security proof of a quantum cryptographic protocol relies on bounding the quantum min entropy of a measurement, conditioned on an adversary's quantum side information, where this bound must be based only on observable statistics (e.g., the noise in the channel).  Much work goes into deriving entropy bounds for various scenarios, with one of the most prominent, and powerful, mathematical tools being \emph{entropic uncertainty relations} \cite{coles2017entropic,wehner2010entropic}.  For a general survey of QKD protocols, the reader is referred to \cite{pirandola2020advances,amer2021introduction}.

This paper investigates the so-called \emph{Extended B92} QKD protocol, introduced originally in \cite{lucamarini2009robust}.  This protocol takes the standard B92 protocol \cite{bennett1992quantum}, which encodes key-bits in the basis choice, as opposed to orthogonal states as is done in BB84, and adds two additional ``test'' states to determine a better bound on the fidelity of the channel.  The extended version of the protocol holds several advantages over the standard B92 system in countering the unambiguous state discrimination attack \cite{duvsek2006quantum,tamaki2009unconditional}.  While these protocols typically do not outperform decoy-state BB84, in practice they can lead to simpler implementations \cite{wang2026performance}.  This makes the study of alternative protocols, such as this, an important endeavor.  

In this work, we derive a new security proof to analyze the protocol's performance in the finite key scenario, against arbitrary, general, attacks.  While this protocol has been analyzed in the asymptotic scenario \cite{lucamarini2009robust}, and also a version of this protocol has been analyzed in the finite key scenario against \emph{collective attacks} only, in \cite{amer2020finite}, to our knowledge our work, here, is the first proof of security for this protocol in the finite key setting against general, coherent, attacks.  Though we restrict ourselves to loss-less channels and single-qubit sources, our methods may be suitably modified to handle other scenarios, as we comment on later.  Deriving a finite key proof of security to handle the single qubit case, is usually a first step towards analyzing more general quantum channels and sources. 


Our work bounds the quantum min entropy directly, without relying on any approximation methods such as the asymptotic equipartition property (AEP) \cite{tomamichel2009fully} as required in prior work.  Interestingly, our proof methodology can be extended to a large class of QKD protocol, where parties reject data based on the outcome of a measurement, and may hold broader interest when analyzing other protocols that involve encoding classical information into non-orthogonal states.  Our proof technique makes use of a quantum sampling framework of Bouman and Fehr \cite{bouman2010sampling}, along with modified proof methods from sampling based entropic uncertainty relations \cite{yao2022quantum}. 



In general, we consider the following scenario: First Eve creates an arbitrary state, sending $N$-qubits to Alice, and $N$-qubits to Bob, while keeping an entangled ancilla.  We do not assume any collective attack structure on the state. A test is performed by Alice and Bob, where they will measure some of the received qubits. This process results in measurement data and a post measured state.  After this, a filtering stage is performed, where Alice and Bob reject some of the remaining signals based on the outcome of some measurement.   Finally, the remaining systems are measured (those that were not rejected, and those which were not used for sampling) and a secret key is distilled.

Filtering like this is common in many QKD protocols.  One must bound the quantum min entropy of the conditional \emph{accepted} state, which may be lower than the entropy in the entire state before filtering, as parties may inadvertently reject signals that Eve had a lot of uncertainty on; Eve may also be able to influence which signals are rejected, giving her greater information in the conditional, accepted state.

Our main result, at a high level, is to show that the final secret key size, of this general protocol, is $\ell$-bits, where:
\begin{equation}
  \ell \approx n_0 \cdot c - \gamma(\mathcal{S}),
\end{equation}
where $n_0$ is the number of accepted signals, $c$ is a function of party measurements, and  $\gamma(\mathcal{S})$ is a function of the underlying \emph{classical} sampling strategy, used for testing the state. Our result is formalized in Theorem \ref{thm:result:main}.  While our result is general, we show how it can be applied to the Extended B92 protocol, in Section \ref{sec:application:b92}.  


To our knowledge the only finite key security proof for the Extended B92 protocol was derived in \cite{amer2020finite}, however that proof assumed collective attacks and did not extend to general attacks.  It also made heavy use of mismatched measurements \cite{barnett1993eavesdropping,watanabe2008tomography}, which add to the sampling burden.  Our Theorem \ref{thm:result:main} works against arbitrary general attacks ``out of the box,'' and derives a bound on the quantum min entropy directly whereas prior work derived a bound on von Neummann entropy, and used approximation methods to promote that analysis to the finite key setting.  We show that not only is our proof, here, more general than prior work, but it also shows higher key rates are possible than previously thought for the Extended B92 protocol, in low-signal scenarios.

Overall, our contributions are to prove security of a general QKD protocol and apply our methods to the Extended B92 protocol.  We show that our results allow for much higher key generation rates for that protocol than previously thought, when the number of signals is low.  Our result also converges to the best-known asymptotic key-rate bound.  Our main result is applicable to other QKD and quantum cryptographic protocols, and our proof methodology may have broad application outside of this work, as it shows a novel method to apply Bouman and Fehr's results from \cite{bouman2010sampling}, to more complicated cryptographic protocols (in this case, those which involve discarding systems based on measurement outcomes which may be influenced by an adversary's control of the quantum channel).

\subsection{Preliminaries}\label{sec:introduction:prelim}

Given a word $q \in \{0,1\}^N$, and a subset $t \subset \{1, 2, \cdots, N\}$, we write $q_t$ to mean the substring of $q$ indexed by $t$ and $q_{-t}$ to mean the substring of $q$ indexed by the complement of $t$.  We write $q_i$ to mean the $i$'th bit of $q$.  Let $wt(q)$ be the Hamming weight of $q$, namely the number of times $1$ (one) appears in $q$, and $w(q) = \frac{1}{N}wt(q)$, which is the relative Hamming weight of $q$.  More generally, let $\num_j(q)$ be the number of times $j$ appears in $q$ for either $j=0,1$ (thus, $\num_1(q) = wt(q)$).

Let $\mathcal{M} = \{\ket{m_0}, \ket{m_1}\}$ be an orthonormal basis; then, for $i=0,1$, we write $\ket{i}^M$ to mean $\ket{m_i}$.  If the superscript is not specified (i.e., $\ket{i}$), we assume the computational $Z$ basis.  Given $q\in\{0,1\}^N$, we write $\ket{q}^M$ to mean $\ket{q_1}^M\cdots \ket{q_N}^M = \ket{m_{q_1}}\cdots \ket{m_{q_N}}$.

Given a density operator $\rho_{AB}$ acting on some Hilbert space $\mathcal{H}_A\otimes\mathcal{H}_B$, we will write $\rho_A$ to mean the result of tracing out $B$.  Similar for three or more systems.  Given a pure state $\ket{\psi}$ we will write $\kb{\psi}$ to mean $\kb{\psi} = \ket{\psi}\bra{\psi}$.

Let $\rho_{AB}$ be a \emph{classical-quantum (cq) state} where the $A$ register is $n$-bits.  Then, the \emph{quantum min entropy} \cite{renner2008security,konig2009operational} is defined to be:
\begin{equation}
\Hmin(A|E)_\rho = -\log_2\max_{\mathcal{E}_a}\sum_aPr(A=a)tr\left(\mathcal{E}_a\rho_{E}\up{a}\right),
\end{equation}
where the maximum is over all POVMs acting on Eve's ancilla, while $\rho_{E}\up{a}$ is Eve's ancilla conditioned on Alice's classical register being $a$.

The \emph{smooth min entropy} \cite{renner2008security} is defined to be $\Hmin^\epsilon(A|E)_\rho = \sup_\sigma \Hmin(A|E)_\sigma$, where the supremum is overall all quantum states $\sigma$ which are $\epsilon$ close to $\rho$ in trace distance, i.e., $\trd{\rho-\sigma} \le \epsilon$.

A useful property of min entropy is the following: Given a mixed state $\rho_{AEZ}$, classical in $Z$, it holds that:
\begin{equation}\label{eq:introduction:min-mix}
\Hmin(A|E)_\rho \ge \Hmin(A|EZ)_\rho \ge \min_z\Hmin(A|E, Z=z)_\rho,
\end{equation}
where $\Hmin(A|E, Z=z)_\rho$ is the min entropy in the state conditioned on $Z$ being a specific value $z$.

Another useful lemma we will use later is the following which was proven in \cite{bouman2010sampling} based on a proof in \cite{renner2008security} (see also \cite{thomas2025new} for more discussion on how the $c$ value appears in this lemma):
\begin{lemma}\label{lemma:introduction:min-sup}
  Let $\ket{\psi}_{AE} = \sum_{a\in J}\ket{a}^M_A\ket{E_a}$ be a quantum state, with $J\subset\{0,1\}^n$.  Assume a measurement of the $A$ system is made in some other orthonormal basis $N$, resulting in quantum state $\rho_{NE}$.  Then:
  \begin{equation}
    \Hmin(N|E)_\rho \ge n\cdot c - \log_2|J|,
  \end{equation}
  where $c = -\log_2\max_{i,j}|\braket{n_i|m_j}|^2$.
\end{lemma}

Quantum min entropy can be used to bound the number of secret, uniform random bits, that may be extracted from a cq-state.  In particular, it was shown in \cite{renner2008security} that, following a privacy amplification process, involving the hashing of $n$-bit register $A$ to an $\ell$-bit register $f(A)$, via a randomly chosen two-universal hash function, $f(\cdot)$, it holds that:
\begin{equation}\label{eq:introduction:pa}
\trd{\rho_{f(A), EF} - \frac{I}{2^\ell}\otimes \rho_{EF}} \le 2^{-\frac{1}{2}(\Hmin^\epsilon(A|E)_\rho - \ell)} + 2\epsilon
\end{equation}
The above is a useful identity for bounding the secret key size of a QKD protocol.  Note that the min entropy computation, on the right-hand side of the above equation, is based on the state before privacy amplification is run.

A QKD protocol is said to be $\epsilon$-secure if \cite{renner2008security}:
\begin{equation}\label{eq:introduction:ep-sec}
p_{ok}\trd{\rho_{KE} - \frac{I}{2^\ell}\otimes\rho_E} \le \epsilon,
\end{equation}
where $p_{ok}$ is the probability that Alice and Bob do \emph{not} abort the protocol.  Above, $\rho_{KE}$ is the state of the system after running the protocol, which includes error correction and privacy amplification, conditioned on not aborting. Here $K$ is the secret key register.

\subsection{Quantum Sampling}
We will use a quantum sampling framework introduced in \cite{bouman2010sampling} by Bouman and Fehr.  We will only briefly summarize the result here.  Consider a \emph{classical sampling strategy}, denoted $\samp$, over words $q\in\{0,1\}^N$ which consists of a distribution $P_T$ over subsets of $\{1, \cdots, N\}$, along with a set of ``guess'' and ``target'' functions, $g_j$ and $\tau_j$ respectively.  Each $g_j,\tau_j:\{0,1\}^*\rightarrow \mathbb{R}$.  The strategy chooses a subset and evaluates $g_j(q_t)$ for all $j$.  Ideally, it should hold that each guess $g_i(q_t)$ is $\delta$-close to a target value on the unobserved portion $\tau_j(q_{-t})$.  Fix $\delta > 0$ and a subset $t\subset\{1,\cdots, N\}$ such that $P_T(t) > 0$ and consider the set:
\begin{equation*}
\mathcal{G}^t_\delta = \{q\in\{0,1\}^N \st \max_j|g_j(q_t)-\tau_j(q_{-t})|\le\delta\}.
\end{equation*}
The above set represents ``good words'' such that if $t$ is the subset chosen, the sampling strategy ``succeeds.''  One is interested in the failure probability of the strategy, namely:
\begin{equation}
\epsilon^{cl}_\delta = \max_{q\in\{0,1\}^N}Pr(q \not \in \mathcal{G}_\delta^t)
\end{equation}
where the probability is over subset choices $t$.

The alphabet need not be bit-strings, and can be more complex elements.  For instance, a \emph{multi-party sampling strategy} is similar to the above, but operates over words $q = (q^A, q^B) \in \{0,1\}^N\times\{0,1\}^N$.  Now, $q_t = (q^A_t, q^B_t)$, while the guess and target functions are $g_j, \tau_j : \{0,1\}^*\times\{0,1\}^* \rightarrow \mathbb{R}$.  This simulates Alice and Bob sampling their respective portions of the word and evaluating a joint function of their individual observations.

A classical sampling strategy, multi-party or otherwise, can be extended to the quantum domain in the following, natural, manner: A state $\ket{\psi}_{ABE}$ is given, where the $A$ and $B$ portions are $N$ qubits each.  Alice and Bob choose $t$ according to the sampling strategy, then measure those qubits, indexed by $t$, in some basis $M = \{\ket{m_0}, \ket{m_1}\}$.  The main result from \cite{bouman2010sampling} is that the post-measured state collapses to a superposition of ``good words'' relative to the given basis $M$.

Formally, let:
\begin{equation}\label{eq:introduction:good-states}
\mathcal{G}^t_\delta(M) = \text{span}(\ket{q}^M \st q \in \mathcal{G}^t_\delta)\otimes \mathcal{H}_E.
\end{equation}
Then the following theorem holds:
\begin{theorem}\label{thm:introduction:sample}
  (From \cite{bouman2010sampling}): Let $\delta > 0$, $M$ an orthonormal basis, and $\ket{\psi}_{ABE}$ be a state as described above.  Then there exist \emph{ideal states} $\{\ket{\phi^t}_{ABE}\}_t$, indexed over subsets $t$, such that $\ket{\phi^t}_{ABE} \in \mathcal{G}^t_\delta(M)$, and:
  \begin{equation}
    \frac{1}{2}\trd{\sum_tP_T(t)\kb{t}\otimes\left(\kb{\psi}_{ABE} - \kb{\phi^t}_{ABE}\right)} \le \sqrt{\epsilon_\delta^{cl}}.
  \end{equation}
\end{theorem}

Finally, one may analyze the entropy in ideal states to derive a bound on the key-rate of a protocol.  In particular, the following lemma will be useful later:
\begin{lemma}\label{lemma:introduction:ideal-sec}
(From \cite{thomas2025new}): Let $\rho_{KE}$ be the result of running a QKD protocol on an input state $\ket{\psi}_{ABE}$.  Let $\Hmin(A|E)_\sigma \ge \gamma$, where $\sigma$ is the result of running the same QKD protocol on ideal states, before privacy amplification, and conditioned on not aborting the protocol.  Then the real protocol is $2^{-\frac{1}{2}(\gamma-\ell)} + 4\sqrt{\epsilon^{cl}_\delta}$ secure according to Equation \ref{eq:introduction:ep-sec}.
\end{lemma}


\section{Extended B92 Protocol}\label{sec:protocol:}

Extended B92, introduced in \cite{lucamarini2009robust}, as its name implies, extends the standard B92 protocol \cite{bennett1992quantum} by adding two non-orthogonal test states.  This was meant to keep some of the benefits of B92 style encoding, in particular non-orthogonal states can help with PNS attacks \cite{huttner1995quantum,lutkenhaus2002quantum}, while also countering, better,  the unambiguous state discrimination attack \cite{duvsek2006quantum,ko2018advanced}.  The protocol, at a high level, involves two types of rounds: Test rounds and Key rounds.  If this is a Test round, Alice will send either $\ket{+}$ or $\ket{-}$ (choosing randomly, though not necessarily with uniform probability) which will be used to test the fidelity of the channel.  If it is a Key round, Alice will send, randomly, one of two non-orthogonal states.  These states are denoted $\ket{\phi_0}$ and $\ket{\phi_1}$, where:
\begin{align}%
  \ket{\phi_0} &= \cos\frac{\theta}{2}\ket{+} + \sin\frac{\theta}{2}\ket{-}\label{eq:prot:sig-state}\\
  \ket{\phi_1} &= \cos\frac{\theta}{2}\ket{+} - \sin\frac{\theta}{2}\ket{-}\notag
\end{align}
We also define $\ket{\bar{\phi}_i}$ to be a vector orthogonal to $\ket{\phi_i}$, namely $\ket{\bar{\phi}_0} = \sin\frac{\theta}{2}\ket{+} - \cos\frac{\theta}{2}\ket{-}$ and $\ket{\bar{\phi}_1} = \sin\frac{\theta}{2}\ket{+} + \cos\frac{\theta}{2}\ket{-}$. 

The receiver, Bob, is allowed to measure in the $X$ basis (for Test rounds) or, on key rounds, is able to measure using POVM $\{M_0, M_1, M_?\}$, where $M_0 = p\kb{\bar{\phi}_1}$, $M_1 = p\kb{\bar{\phi}_0}$ and $M_? = I - M_0 - M_1$.  Here $p = \frac{x^2}{2\cos^2\frac{\theta}{2}}$ for some $x$ depending on the measurement devices, with $x = 1$ modeling ideal devices, and $x = \cos\frac{\theta}{2}$ modeling practical devices \cite{lucamarini2009robust}.  A measurement of $M_?$ is inconclusive and will lead to Bob discarding that round.  

On Test rounds, Alice and Bob will disclose their choices and measurement outcomes.  On Key rounds, Bob will only disclose whether he observed $M_?$ or not.  $M_?$ will indicate an \emph{inconclusive measurement} and any such rounds are discarded.  Otherwise, his observation of $M_k$ will lead directly to his raw key bit of $k$.  It is not difficult to see that this protocol is exactly B92 but with the extension of two non orthogonal states during the Test rounds.


The above can be reduced to an equivalent entanglement based protocol as discussed in \cite{lucamarini2009robust} using the following identity:
\begin{equation}\label{eq:application:ent-eq}
  \frac{1}{\sqrt{2}}\ket{0,\phi_0}_{AB} + \frac{1}{\sqrt{2}}\ket{1,\phi_1}_{AB} = \cos^2\frac{\theta}{2}\ket{++} + \sin^2\frac{\theta}{2}\ket{--},
\end{equation}

Specifically, the entanglement based version that we analyze operates as follows:
\heading{Public Input}
\begin{itemize}%
\item $N$: Total number of rounds of the protocol.
\item $m \le N/2$: Total number of test rounds (see discussion, below).
\item $\theta$: Signal state parameter (see Equation \ref{eq:prot:sig-state})
\item $Q$: Maximal tolerated error, after which parties abort the protocol.
\item $n_0$: Minimal number of tolerated accepted rounds.  If fewer rounds than this lead to a conclusive measurement outcome, parties abort.
\end{itemize}

\heading{Quantum Communication Stage}
\begin{enumerate}%
\item A source (potentially an adversary) produces a quantum state $\ket{\psi}_{ABE}$, where the $A$ and $B$ registers consist of $N$ qubits each, while the $E$ system is arbitrary, but finite.  The $A$ qubits are sent to Alice and the $B$ qubits are sent to Bob.  Ideally, this should be an $N$-fold tensor product of states of the form in Equation \ref{eq:application:ent-eq}.
\item Alice and Bob agree on a random subset $t \subset\{1,\cdots, N\}$ such that $|t| = m$.  See text, below for discussion on this.  Note, we are assuming measurement devices are memory-less for this test, as is also done in \cite{tomamichel2012tight}.
\item For those rounds $i\in t$, Alice and Bob measure their respective qubits in the $X$ basis, reporting their outcomes which we denote $q^A$ and $q^B$ (which are $m$-bit strings).  This allows them to compute $q = w(q^A\oplus q^B)$, namely the relative number of errors in their $X$ basis outcomes.  If $q > Q$, Alice and Bob abort.  As seen in Equation \ref{eq:application:ent-eq}, their outcomes should, ideally, be correlated.
\item On the remaining $n = N-m$ rounds, Bob will measure using POVM $\{M_0, M_1, M_?\}$, as discussed above.  Alice will measure in the $Z$ basis.
\item Bob will disclose all rounds where he observed $M_?$.  Let $c_0$ be the number of rounds which were not discarded (namely, the number of rounds where his measurement outcome was $M_0$ or $M_1$).  If $c_0 < n_0$, parties abort.
\item For all rounds that were not discarded, Alice's raw key will be her $Z$ basis measurement outcome, while Bob's will be his outcome $M_k$, for $k\in\{0,1\}$.
\end{enumerate}

Following the conclusion of the above, parties will take their raw keys and perform error correction, followed by privacy amplification, as normal in QKD.  We comment that the method of choosing the subset $t$, above, can be done by having Alice choose a random subset and sending it to Bob.  Alternatively, as in practice, Alice and Bob will choose independently which rounds will be Test rounds and then the subset $t$ is constructed from those rounds where both parties choose Test.  We analyze the case where $t$ is chosen completely by one party; however our analysis works in the case where both parties choose independently at random, though one must take into account that $m$, the size of the subset, is a random variable, and parties should abort if it is ``too small.''  However, these details are purely classical sampling details that are easily added to our analysis.

To analyze the security of the above protocol, we will need a bound on the min entropy of Alice's $Z$ basis measurement, conditioned on Eve, for all rounds that were not discarded (in order to apply Equations \ref{eq:introduction:pa} and \ref{eq:introduction:ep-sec}).  Note that Bob's final raw key result does not matter for this computation (it will matter for correctness of the protocol, of course, and for determining a bound on the error correction leakage).  For this reason, we will actually consider a ``toy'' version of the protocol, where only Alice gets a raw key, and Bob, following his POVM measurement, will ``shut down.''  Namely, the protocol is identical to the above, except that on Step 4, Bob will measure using a Filtering POVM $\{F_0, F_1\}$, where $F_0 = M_0 + M_1$ and $F_1 = M_?$.  For every round where his filter produces an outcome of $F_1$, he will later signal to Alice to discard that round.  Otherwise the round is kept.  All rounds which are discarded are traced out, and Bob's remaining system is also traced out.  Alice measures the non-discarded rounds in the $Z$ basis.  It is not difficult to see that the resulting density operator will be identical (after tracing out Bob for both the real and toy protocol) and, thus, analyzing the toy protocol will produce a valid entropic uncertainty result for the real protocol.


\section{Main Technical Result}\label{sec:result:}

We now turn to our main result.  For this, we consider a very general experiment (which models a QKD protocol, but can also model other cryptographic protocols):

1. On input a quantum state $\rho\up{0}_{ABE}$, produced potentially by Eve who holds the $E$ system, where the $A$ (Alice) and $B$ (Bob) systems are $N$ qubits each, Alice and Bob run a multi-party sampling strategy $\samp$ where all subsets are of size $m$, with respect to orthonormal basis $M=\{\ket{m_0}, \ket{m_1}\}$, to get sampling data $(t,s)$ and some post measured state $\rho\up{t,s}_{ABE}$ where, now, the $A$ and $B$ systems are $n = N-|t| = N-m$ qubits each.  As a mixed state, this is
  \begin{equation}
    \rho_{ABETS} = \sum_{t,s}p(t,s)\kb{t,s}\otimes\rho_{ABE}\up{t,s}
  \end{equation}
  Note that the sampling data $s$ may consist of numerous entries, depending on the given sampling strategy.

2. Bob now measures his unsampled qubits (in the new $B$ register) using measurement operators $\filter^B = \{F^B_0, F^B_1\}$. Alice measures her $A$ system using $\filter^A = \{F_0^A, F_1^A\}$.  These act as ``filtering'' measurements where a result of ``$1$'' will mean to discard that particular system/round.  The post-measured state of these operators is also saved in the new $A$ and $B$ registers (which are still $n$ qubits each).  Let $D^B$ and $D^A$ be the (classical) registers storing the outcome of these measurements and let $D$ be the register such that $D_i = 0$ only if both $D^A_i=0$ and $D^B_i=0$ (otherwise $D_i = 1$).  Parties will later discard any qubit where $D_i = 1$.  Note that, in practice, data discarding and filtering may be done by first measuring in a final basis, then sifting through their results; however this can be modeled as first applying a suitable filtering measurement as we do here (e.g., the measurement may project into a subspace of states that would have been discarded or accepted).

  3. Parties apply an \emph{Abort map} $\mathcal{R}_\mathcal{S}$, which will set an abort flag in register $R$ to ``$1$'' (i.e., \texttt{True}), if $s \not\in \mathcal{S}$ or $\num_0(D) < n_0$ for user specified $\mathcal{S}$ and $n_0$.  The set $\mathcal{S}$ can specify, for instance, the maximal tolerated noise parties will accept before aborting, while $n_0$ is the user-specified minimum allowed number of accepted (not discarded) rounds.

  4. Alice measures her remaining systems (those not rejected by the filtering measurements) in the $Z$ basis to get register $A_Z$.  Bob measures in some other two-outcome POVM to get register $B_P$.  These are their raw keys.  Note that Alice could measure in an alternative basis in an actual protocol, however we can model that here simply by adding a change of basis operation to Alice's filtering measurements. The resulting density operator is denoted $\rho_{A_ZB_PETSD^AD^BDR}$.

  5. Assuming the abort flag is not set, parties perform error correction (EC), leaking at most $\leakEC$ bits and finally privacy amplification (PA), hashing the resulting raw key registers (the error corrected $A_Z$ and $B_P$ registers) to $\ell$-bits.

Our main result is to show that the min entropy in the $A_Z$ register, before error correction and privacy amplification, but after discarding systems, is ``high,'' or at least bounded by a function of $\mathcal{S}$, $n_0$ and the classical strategy $\samp$.  In particular, consider the following function:
\begin{align}
  &\gamma(\samp, \mathcal{S}, c_0) = \max_{\substack{s\in\mathcal{S}\\d\in\{0,1\}^n:\num_0(d)= c_0\\b\in\{0,1\}^n\\a\in\{0,1\}^{n-c_0}}}\label{eq:gamma}\\
  &\log_2\left|\left\{q\in\{0,1\}^{c_0} \st \max_j\left|s_j-\tau_j\left(\pi_d\left(q, a\right), b\right)\right| \le \delta\right\}\right|.\notag
\end{align}
where, above, $\pi_d:\{0,1\}^{\num_0(d)}\times\{0,1\}^{\num_1(d)}\rightarrow\{0,1\}^n$ is a permutation that places the first input into the those bits of the output string where $d$ is zero and places the second input to those bits of the output string where $d=1$.  For example, if $d = 01011$, then:
\begin{equation}\label{eq:result:pi-example}
\pi_d(ab, cde) = acbde.
\end{equation}

Our main result, below, shows that if one can bound the above function, then one can derive a bound on the quantum min entropy of Alice's measurements on those systems not discarded.  Bounding the above function will depend on the sampling strategy; for many, however, it turns out that the set behaves nicely, as we show in Section \ref{sec:application:b92}.  For example, a common sampling function is the Hamming weight, which is permutation invariant, and thus simplifies the above expression.  The above description of the function, however, works for any multi-party sampling strategy (and thus any protocol that can be modeled by such a strategy and the above described experiment).

Our main result, then, is stated in the following theorem:
\begin{theorem}\label{thm:result:main}
  Let $\delta > 0$ and $\rho_{ABE}$ be a density operator where the $A$ and $B$ registers consist of $N$ qubits.  Let $\rho_{A_ZB_PETSD^AD^BDR}$ be the result of running the above described protocol (before EC and PA are run).  Then, if for all $j,k\in\{0,1\}$ and $u\in\{A,B\}$ it holds that:
  \begin{equation}\label{eq:result:filter-hyp}
    F_j^u\ket{k}^M = \lambda_u(j|k)\ket{k}^{\wt{M}}
  \end{equation}
  for $\lambda_u(j|k) \in \mathbb{C}$, and some other (or same) orthonormal basis $\wt{M}$ (where, recall, $M$ is the sampling basis), then it holds that a $5\sqrt{\epsilon_\delta^{cl}}$-secure key may be distilled from the above state of length $\ell$ with:
  \begin{equation}
    \ell = \min_{c_0 \ge n_0}\big[c_0\cdot c - \log_2\gamma(\samp, \mathcal{S}, c_0)\big] - \leakEC - 2\log_2\frac{1}{\sqrt{\epsilon_\delta^{cl}}},
  \end{equation}
where $c = -\log_2\max_{i,j}|\braket{\wt{m}_i|j}|^2$.
\end{theorem}
\begin{proof}
  Our proof proceeds in four steps.  First, we will use Theorem \ref{thm:introduction:sample} to construct ``ideal'' states, our goal being to analyze these and then use Lemma \ref{lemma:introduction:ideal-sec} to promote the analysis to the real state.  Next, steps two and three involve tracing the protocol's execution, including all filtering operations and measurements, on the ideal states.  Finally, we show how to bound the min entropy of the resulting state, as a function of $\gamma(\samp, \mathcal{S}, c_0)$.
  
  \textbf{Step 1, Ideal State Construction:}
  First, consider a pure input state $\rho\up{0}_{ABE} = \kb{\psi}_{ABE}$. If the input state is not pure, we may purify it and give the purification system to Eve which can only be to her benefit.  By Theorem \ref{thm:introduction:sample}, there exist ideal states $\{\ket{\phi\up{t}}\}$, indexed by subsets $t$, such that:
  \begin{equation}
    \frac{1}{2}\trd{\sum_tP_T(t)\kb{t}\otimes\left(\kb{\psi}_{ABE} - \kb{\phi\up{t}}_{ABE}\right)}\le \sqrt{\epsilon^{cl}_\delta},
  \end{equation}
  and where each $\ket{\phi\up{t}} \in \mathcal{G}^t_\delta(M)$, where this subspace is induced by the given sampling strategy, as discussed in Section \ref{sec:introduction:prelim}.

  We trace the execution of the protocol above on the ideal system.
  After the sampling strategy runs, the ideal system is in the mixed state:
 $   \sum_tP_T(t)\kb{t}\otimes\sum_{s}p(s|t)\kb{s}\otimes\kb{\phi\up{t,s}},$
  where the second sum is over all possible outputs of the sampling strategy, $s$, for this input state (which is a finite sum) and where:
   $ \ket{\phi\up{t,s}} = \sum_{(q^A,q^B)\in J_s}\ket{q^A,q^B}^M\otimes \ket{E_{q^Aq^B}^{t,s}}.$
  Here, we have $J_s = $
  \begin{equation*}
    \{(q^A,q^B)\in\{0,1\}^n\times\{0,1\}^n \st \max_j|s_j - \tau_j(q^A,q^B)|\le\delta\}.
  \end{equation*}
  The above follows, since $\ket{\phi\up{t}}\in\mathcal{G}_\delta^t(M)$ (defined in Equation \ref{eq:introduction:good-states}).

  \textbf{Step 2, Application of Filtering Measurements:}
  Bob now applies his filter measurement $\filter^B$.  Similarly Alice measures using her filtering measurement.  Storing the resulting measurement outcomes in registers $D^AD^B$ yields the mixed state:
  \begin{align*}
    &\sum_{t,s}p(t,s)\kb{t,s}\sum_{d^A,d^B\in\{0,1\}^n}\kb{d^Ad^B}_{D^AD^B}\\
    &\otimes P\left(\sum_{(q^A,q^B)\in J_s}{F^A_{d^A}}\ket{q^A}^M{F^B_{d^B}}\ket{q^B}^M\ket{E^{t,s}_{q^Aq^B}}\right),
  \end{align*}
  where $P(\ket{z}) = \kb{z}$.
  Above, by ${F_{d^A}^A}\ket{q^A}^M$ we mean $F_{d^A_1}^A\ket{q^A_1}^M\otimes\cdots\otimes{F_{d^A_n}^A}\ket{q^A_n}^M$.  Similarly for Bob's filtering operation.  From our theorem hypothesis, Equation \ref{eq:result:filter-hyp}, we can write the above as:
  \begin{align*}
    &\sum_{t,s}p(t,s)\kb{t,s}\sum_{d^A,d^B\in\{0,1\}^n}\kb{d^Ad^B}_{D^AD^B}\otimes\\
    &P\left(\sum_{(q^A,q^B)\in J_s}\lambda_A(d^A|q^A)\lambda_B(d^B|q^B)\ket{q^A,q^B}^{\tilde{M}}\ket{E^{t,s}_{q^Aq^B}}\right),
  \end{align*}
  where $\lambda_A(d^A|q^A) = \lambda_A(d^A_1|q^A_1)\times\cdots\times\lambda_A(d^A_n|q^A_n)$ (and, of course, similarly for Bob).  Note that the pure state within the projector function $P(\cdot)$ is not necessarily normalized and its inner-product represents the probability of the filtering operation producing that particular value of $d^A$ and $d^B$.

  Setting the $D$ register appropriately (where $D_i = 0$ only if $D_i^A = D_i^B = 0$, namely $D = D^A\vee D^B$ where $\vee$ is the bitwise logical OR operation) yields:
  \begin{align*}
    &\sum_{t,s}p(t,s)\kb{t,s}\sum_{d\in\{0,1\}^n}\kb{d}_D\sum_{\substack{d^A,d^B\\d^A\vee d^B = d}}\kb{d^Ad^B}_{D^AD^B}\otimes\\
    &P\left(\sum_{(q^A,q^B)\in J_s}\lambda_A(d^A|q^A)\lambda_B(d^B|q^B)\ket{q^A,q^B}^{\tilde{M}}\ket{E^{t,s}_{q^Aq^B}}\right),
  \end{align*}

  \textbf{Step 3, Final Raw Key Measurements:}
  It is at this point that parties will run the remainder of the protocol.  Namely, for those systems not indexed for discarding (i.e., those where $D_i = 0$), Bob will measure in his key distillation POVM and Alice will measure in the $Z$ basis, leading to her raw key.  Since we are only interested in the entropy of Alice's measurement, we trace Bob out.  Equivalently, we may first trace out Bob's system and then discard Alice's qubits where $D_i=1$ and finally measure the remaining $A$ systems in the $Z$ basis.  Before this final measurement and discarding of Alice's system, but after tracing out Bob's entire register, we have:
  \begin{align*}
    &\sum_{t,s}p(t,s)\kb{t,s}\sum_{d\in\{0,1\}^n}\kb{d}_D\sum_{\substack{d^A,d^B\\d^A\vee d^B = d}}\kb{d^Ad^B}_{D^AD^B}\\
    &\otimes\sum_{q^B\in\{0,1\}^n}\lambda_B^2(d^b|q^B)\kb{\nu(t,s,d,d^A,d^B,q^B)}_{AE},
  \end{align*}
  where $\ket{\nu(t,s,d,d^A,d^B,q^B)}_{AE} = $
  \begin{equation}
 \sum_{q^A\in J(s,q^B)}\lambda_A(d^A|q^A)\ket{q^A}^{\tilde{M}}\ket{E^{t,s}_{q^Aq^B}}
  \end{equation}
  and:
  \begin{equation*}
    J(s,q^B) = \{q^A\in\{0,1\}^n \st \max_j|s_j - \tau(q^A,q^B)| \le \delta\}.
  \end{equation*}
  Note that the $\ket{\nu(t,s,d,d^A,d^B,q^B)}$ states are sub-normalized.

  Now we will trace through Alice's operations on the above state.  First, she traces out those systems where $D_i = 1$.  Equivalently she measures them and discards the output.  To maintain the dimension of the system, she replaces any discarded system with a $\ket{0}$.  We will follow the protocol's execution on a particular $\ket{\nu(t,s,d,d^A,d^B,q^B)}$ state (i.e., conditioning on this particular outcome); the joint mixed state will then simply be a weighted sum of these outputs.

  To trace this part of the protocol, instead of summing over $q^A\in\{0,1\}^n$, we will, for a particular $d$, write $q^A = \pi_d(q\up{A,0}, q\up{A,1})$, where $\pi_d$ is the permutation discussed earlier (see Equation \ref{eq:result:pi-example}).  We will then sum over these sub strings, allowing us to write the state as:
  \begin{align*}
    &\sum_{q\up{A,1}\in\{0,1\}^{\num_1(d)}}\lambda_A^2(1\cdots 1|q\up{A,1})\underbrace{\kb{0\cdots0}}_{\num_1(d)\text{ times}}\\
    &\otimes\kb{\mu(t,s,d,d^A,d^B,q^B,q\up{A,1})}_{AE},
  \end{align*}
  where $\ket{\mu(t,s,d,d^A,d^B,q^B,q\up{A,1})}_{AE} = $
  \begin{equation}\label{eq:result:mu}
    \sum_{q\up{A,0}}\lambda_A(0\cdots 0|q\up{A,0})\ket{q\up{A,0}}^{\tilde{M}}\ket{E^{t,s}_{\pi_d(q\up{A,0},q\up{A,1})q^B}}
  \end{equation}
  and where the sum is over all  $q\up{A,0}\in J(s,q^B,q\up{A,1},d) = $
  \begin{equation}\label{eq:result:Jset2}
    \{x\in\{0,1\}^{\num_0(d)} \st \max_j|s_j - \tau_j(\pi_d(x,q\up{A,1}), q^B)| \le \delta\}.
  \end{equation}
  Note that some of the $\ket{\mu(t,s,d,d^A,d^B,q^B,q\up{A,1})}$ vectors may be the zero vector.

  At this point, we are at step 3 of the protocol where parties set an ``abort'' flag if $s \not\in \mathcal{S}$ or if $\num_0(d) < n_0$.  Conditioned on not aborting, the system (we now combine everything again) collapses to the mixed state:
  \begin{align*}
    &\sigma\up{ok} = \frac{1}{p_{ok}}\sum_t\sum_{s\in\mathcal{S}}p(t,s)\kb{t,s} \sum_{d\st\num_0(d) \ge n_0}\kb{d}\\
                    &\otimes\sum_{\substack{d^A,d^B\\\st d=d^A\vee d^B}}\kb{d^A,d^B}\sum_{\substack{q^B\in\{0,1\}^n\\ q\up{A,1}\in\{0,1\}^{\num_1(d)}}}p(q\up{A,1},q^B,d^A,d^B)\\
    &\otimes\kb{0\cdots 0}\kb{\tilde{\mu}(t,s,d,d^A,d^B,q^B,q\up{A,1})},
  \end{align*}
  where the above scalars $p(q\up{A,1},q^B,d^A,d^B)$, can be easily derived, though their exact form is not important to the proof.  Furthermore, the state $\ket{\wt{\mu}(\cdots)}$ is the normalized version of $\ket{\mu(\cdots)}$ (from Equation \ref{eq:result:mu}).

  \textbf{Step 4, Final Entropy Bound:}
  Alice will now measure her non-discarded systems in the $Z$ basis resulting in her raw key (a register denoted $A_Z$).  From Equation \ref{eq:introduction:min-mix}, we have: $\Hmin(A_Z|ETSDD^AD^B)_{\sigma\up{ok}}$
  \begin{align}
       &\ge \min_{\substack{s\in S\\t\\q^B\in\{0,1\}^n\\d\in\{0,1\}^n \st \num_0(d) \ge n_0\\d^A,d^B\st d = d^A\vee d^B\\q\up{A,1}\in\{0,1\}^{\num_1(d)}}}\Hmin(A_Z|E)_{\widetilde{\mu}(t,s,d,d^A,d^B,q^B,q\up{A,1})}.
  \end{align}

  By Lemma \ref{lemma:introduction:min-sup}, we have, for every $V=(t,s,d,d^A,d^B,q^B,q\up{A,1})$, it holds that
  \begin{equation}
    \Hmin(A_Z|E)_{\widetilde{\mu}(V)} \ge \num_0(d)\cdot c - \log_2|J(s,q^B,q\up{A,1}, d)|,
  \end{equation}
  where $J(s,q^B,q\up{A,1}, d)$ is defined in Equation \ref{eq:result:Jset2}.
  Note this does not depend on the specific value of the individual $d^A$ and $d^B$, but instead only the joint value $d$ (which, ultimately, is the bit-wise OR of both individual values as discussed above).
  Above, we have $c = -\log_2\max_{i,j}|\braket{\wt{m}_i|j}|^2$ as described in the theorem statement.  This allows us to conclude:
  \begin{equation}
    \Hmin(A|_ZETSDD^AD^B) \ge \min_{c_0\ge n_0}\left(c_0\cdot c - \log\gamma(\samp, \mathcal{S}, c_0)\right).
  \end{equation}
  This completes the analysis of the ideal state, conditioned on not aborting the protocol.  By Lemma \ref{lemma:introduction:ideal-sec}, if we set the privacy amplification size to $\ell=\min_{c_0\ge n_0}\left( c_0\cdot c - \log\gamma(\samp, \mathcal{S}, c_0)\right) - 2\log\frac{1}{\sqrt{\epsilon^{cl}_\delta}}$, then the resulting secret key will be $5\sqrt{\epsilon^{cl}_\delta}$-secure according to Lemma \ref{lemma:introduction:ideal-sec}.  Of course, we must still take into account leakage due to error correction.  We may assume this is part of Eve's system, and suitably deduct from our min entropy bound above using the chain rule of min entropy \cite{renner2008security}.  This allows us to set the secret key size to $\ell = \min_{c_0\ge n_0}\left( c_0\cdot c - \log\gamma(\samp, \mathcal{S}, c_0)\right) - \leakEC - 2\log\frac{1}{\sqrt{\epsilon^{cl}_\delta}}$ as desired.
\end{proof}

We comment that, in practice, a correctness check can also be performed which would deduct an additional $\log\frac{1}{\epsilon_{cor}}$ bits from the final secret key \cite{tomamichel2012tight} where $\epsilon_{cor}$ is the desired, maximal, failure probability of error correction.  However we do not go into that detail here as it is a trivial addition to our main result above and does not deduct substantially from the final result.

We also comment that our requirement on the filtering measurements, Equation \ref{eq:result:filter-hyp}, may seem strong at first, however many practical data discarding techniques can be modeled by such a system.  For instance, protocols which involve Alice and Bob measuring in different orthonormal bases (those qubits not sampled in $t$), and based on the results, discarding outcomes.  



\section{Application: Extended B92}\label{sec:application:b92}
As an application, we consider the so-called \emph{Extended B92} protocol, originally introduced in \cite{lucamarini2009robust}. The first, and to our knowledge only, finite key proof of this type of protocol was derived in \cite{amer2020finite}, which derived a finite key proof of a simplified version of the protocol (where only one Test state was used), but only for collective attacks.  To our knowledge there is no finite-key security proof assuming general, coherent attacks, for either the original protocol from \cite{lucamarini2009robust}, or the simplified one considered in \cite{amer2020finite}.  In this section, we use our Theorem \ref{thm:result:main} to analyze this protocol; we also compare to prior work, and show our result converges to the asymptotic upper bound in \cite{lucamarini2009robust}, while also giving better results than prior work in \cite{amer2020finite} for small signal sizes.

As discussed in Section \ref{sec:protocol:}, we actually analyze the ``toy'' protocol where, first, Bob will apply a filtering measurement $F^B_0 = \sqrt{M_0+M_1}$ and $F_1^B = \sqrt{M_?}$, rejecting those rounds which cause an observation of $F_1^B$, followed by taking the remaining systems and measuring in the $X$ basis.  Bob is then discarded from the resulting density operator.  This can be modeled in our framework of Section \ref{sec:result:}, by having Bob measure in any basis, at the end, and then tracing out his system.

It can be shown that this filtering operation satisfies the requirements of our Theorem \ref{thm:result:main}.  In particular, since $\ket{\bar{\phi_j}} = \sin\frac{\theta}{2}\ket{+} + (-1)^{1+j}\cos\frac{\theta}{2}\ket{-}$ (see Section \ref{sec:protocol:}), we have:
\begin{align*}
  &M_0 + M_1 = p\kb{\bar{\phi}_1} + p\kb{\bar{\phi}_0}\notag\\\notag\\
  &= pP\left(\sin\frac{\theta}{2}\ket{+} + \cos\frac{\theta}{2}\ket{-}\right)\notag\\
  &+ pP\left(\sin\frac{\theta}{2}\ket{+} - \cos\frac{\theta}{2}\ket{-}\right)\notag\\\notag\\
  & = 2p\left(\sin^2\frac{\theta}{2}\kb{+} + \cos^2\frac{\theta}{2}\kb{-}\right),
\end{align*}
where, recall, $P(\ket{z}) = \kb{z}$.
Thus:
\begin{align}
  F_0^B &= \sqrt{M_0 + M_1} =\sqrt{2p}\left( \sin\frac{\theta}{2}\kb{+} + \cos\frac{\theta}{2}\kb{-}\right).
\end{align}
This of course implies that $F_0^B\ket{x}^X = \lambda(0|x)\ket{x}^X$ for a scalar $\lambda(0|x)$.  It is easy to verify that $F_1^B = \sqrt{M_?} = \sqrt{I-M_0-M_1}$ also satisfies the theorem statement.


Now that we have a filtering operation which correctly models the protocol, we next need a classical sampling strategy in order to employ our main result in Theorem \ref{thm:result:main}.  However, this is straight forward: Alice and Bob will choose subset $t$ and measure in the $X$ basis, reporting their outcomes and computing the number of errors between their observations.  Thus, $g(q_t) = w(q^A_t\oplus q^B_t)$ and $\tau(q_{-t}) = w(q^A_{-t}\oplus q^B_{-t})$ (there is only one guess/target function pair for this protocol).  We will assume $P_T$ chooses subsets of size $|t|=m<N/2$, uniformly at random from the $N$ rounds (where $m$ is fixed, and given by the user).  This sampling strategy was analyzed in \cite{yao2022quantum} and the error probability was found to be:
\begin{equation}\label{eq:application:cl-error}
  \epsilon^{cl}_\delta = 2\exp\left(\frac{-\delta^2m(n+m)}{m+n+2}\right).
\end{equation}

We set $\mathcal{S}$ to represent the maximum allowed $X$-basis noise that users will tolerate before aborting the protocol (maximal $w(q^A_t\oplus q^B_t)$).  Let $Q$ be this maximum allowed noise.  Also let $n_0$ be the minimum number of non-discarded rounds allowed by users before they abort.  We need to determine a bound on $\gamma(\samp, Q, c_0)$ (defined in Equation \ref{eq:gamma}).  Note that due to the structure of the target function, we can simplify this function to:
\begin{equation*}
\max_{\substack{s \le Q\\b\in\{0,1\}^n\\a\in\{0,1\}^{n-c_0}}}\log_2\left|\left\{q\in\{0,1\}^{c_0} \st |s - w((a||q) \oplus b)| \le \delta \right\}\right|,
\end{equation*}
where $a||q$ is the concatenation of strings $a$ and $q$.  The above simplification is due to the fact that we are maximizing over all possible $b$ and $a$ and that permuting bits within this particular target function in both coordinates will not alter it.

Fix $s, b,$ and $a$.  Note that $w((a||q)\oplus b) = \frac{1}{n}(wt(a\oplus b_L) + wt(q\oplus b_R))$, where $b_L$ is the left most $n-c_0$ bits of $b$ and $b_R$ is the right most $c_0$ bits.  By some manipulation and the well known bound on the volume of a Hamming ball, we can write this as:
\begin{align*}
  &\left|\left\{q\in\{0,1\}^{c_0} \st |s - w((a||q) \oplus b)| \le \delta \right\}\right|\notag\\
  &\le \left|\left\{q\in\{0,1\}^{c_0} \st w((a||q) \oplus b) \le s + \delta \right\}\right|\le\notag\\
  &\left|\left\{q\in\{0,1\}^{c_0} \st w(q\oplus b_R) \le \frac{n}{c_0}(s + \delta) - \frac{1}{c_0}wt(a \oplus b_L) \right\}\right|\\
  &\le 2^{c_0h(\frac{n}{c_0}(s + \delta)-\frac{1}{c_0}wt(a\oplus b_L))} \le 2^{c_0h(\frac{n}{c_0}(s +\delta))}
\end{align*}
From this, we conclude $\gamma(\samp,Q,c_0) \le c_0h(\frac{n}{c_0}(Q +\delta))$. Thus, using our Theorem \ref{thm:result:main}, we conclude the secret key size is:
\begin{equation*}\label{eq:application:fin-key}
\ell = \min_{c_0 \ge n_0}c_0\left(1 - h\left(\frac{n}{c_0}(Q+\delta)\right)\right) - \leakEC - 2\log\frac{1}{\sqrt{\epsilon_\delta^{cl}}}.
\end{equation*}
It is easy to see that the above is minimized when $c_0 = n_0$ (i.e., the smallest possible value for $c_0$ before parties abort).

The above expression is valid for any arbitrary quantum channel or attack.  To evaluate, however, we will assume depolarization noise - a common case in evaluating key-rates, and one which will allow us to readily compare our key-rate bound with prior work.  Such a channel maps a qubit density operator $\rho$ to $\mathcal{E}_Q(\rho) = (1-2Q)\rho + QI/2$, where $Q$ is the depolarizing parameter.  Keep in mind, however, our security result, Equation \ref{eq:application:fin-key}, works for any channel.  We use this depolarization channel, as is done in the majority of theoretical QKD research, in order to evaluate and compare to prior work when possible.

Using this, we see that the expected $X$ basis error rate will simply be $Q$, while the expected value of $\frac{n_0}{n}$ (i.e., the ratio of accepted rounds to total rounds), is readily found to be:
\begin{equation}
p_a = \frac{n_0}{n} = 4p\alpha^2\beta^2(1-2Q) + 2pQ,
\end{equation}
where we set $\alpha=\cos\frac{\theta}{2}$ and $\beta=\sin\frac{\theta}{2}$.  The above identity is easily found, by tracing the protocol state transmission through the depolarizing channel, and working out the probability that Bob observes $I-M_?$.

For error correction, we will set $\leakEC = n_0h(Q_Z+\delta)$, where $Q_Z$ is the expected raw key error rate, which is readily seen to be (again, by tracing the protocol under this depolarizing noise map):  $Q_Z = pQ/p_a$.

Let $\epsilon > 0$ be the desired security level, specified by the user (later, in our evaluations, we use $\epsilon = 10^{-12}$); then we set $\delta$ to be:
\begin{equation}
\delta = \sqrt{\frac{m+n+2}{m(m+n)}\ln\frac{50}{\epsilon^2}},
\end{equation}
in which case from Equation \ref{eq:application:cl-error} it will hold that $5\sqrt{\epsilon_\delta^{cl}} = \epsilon$ and our secret key will be $\epsilon$-secure by our Theorem \ref{thm:result:main} and Equation \ref{eq:introduction:ep-sec}.

Our results are shown in Figure \ref{fig:application:fig1}, with ideal ($x=1$) and practical ($x=\cos\frac{\theta}{2}$) devices.  We also compare with asymptotic results from \cite{lucamarini2009robust} and note our result converges to these asymptotic results in prior work. As we are the first, to our knowledge, to prove a finite key result for the full version of this protocol \emph{under general attacks}, we do not have other finite key evaluations to directly compare to.  In Figure \ref{fig:application:fig2}, we compare with finite key results from \cite{amer2020finite}, however it is difficult to make a direct comparison as that reference assumed weaker collective attacks and did not handle general attacks as we do (thus, key-rates from \cite{amer2020finite} may be artificially high).  That paper also utilized mismatched measurement bases which collect substantial information about an adversary's attack, which we did not consider here.  On the other, hand, the protocol in \cite{amer2020finite} was also simpler in the sense that only one of the two possible Test states were sent.  Thus, a direct comparison is difficult, however it is the closest finite key result to our work.

We see that our result outperforms this prior work at small signal count, while prior work outperforms, slightly, in a higher number of signals (though both results converge, asymptotically).  Whether our proof can be improved in higher signal counts, or if this is due to the fact that we are considering a stronger security model, remains an open question.

Finally, in Table \ref{tab:orgedcb091}, we evaluate the maximal tolerated error rates, for various $\theta$ and signal counts.


\begin{figure}
  \centering
  \includegraphics[width=0.9\linewidth]{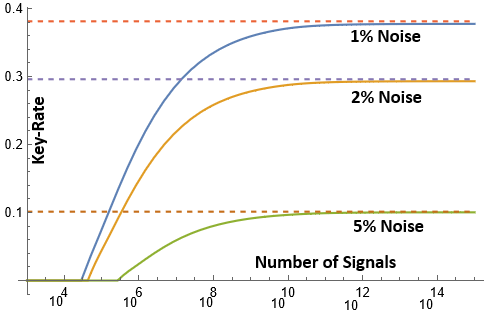}
  \includegraphics[width=0.9\linewidth]{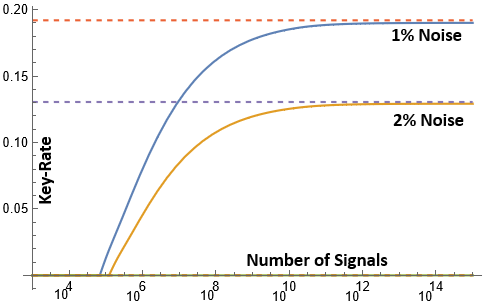}
  \caption{Evaluating our finite key-rate result $\ell/N$ (Solid), where $N$ is the total number of signals sent, and comparing to asymptotic results from \cite{lucamarini2009robust} (Dashed).  Top: Ideal devices ($x=1$ in the POVM); Bottom: practical devices ($x=\cos^2\frac{\theta}{2}$ for the measurement POVMs).  Note the difference in $y$-axis scale between the Top and Bottom graphs.  Here, we test $\theta = \pi/3$ for various noise levels, $Q$.  Similar results are found for other $\theta$, with decreasing key-rates as $\theta$ decreases (which is a property of this protocol \cite{lucamarini2009robust}).}\label{fig:application:fig1}
\end{figure}


\begin{figure}
  \centering
  \includegraphics[width=0.9\linewidth]{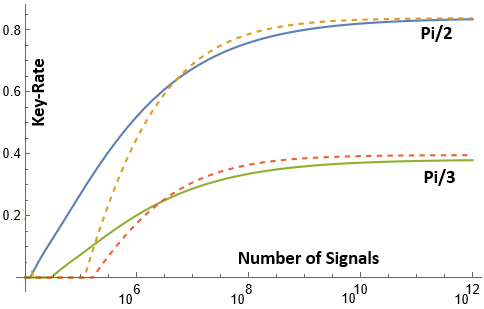}
  \caption{Comparing our new result (Solid) with finite key results in \cite{amer2020finite} (Dashed) for $\theta = \pi/2$ and $\theta = \pi/3$ with ideal measurement devices.  Here, the noise parameter is $Q=1\%$. We note our result gives better key-rates in lower signal counts; similar trends were found in the practical device setting and other noise levels.  See text for additional discussion on how to compare.}\label{fig:application:fig2}
\end{figure}

\begin{table}[h]
\caption{Showing the maximal tolerated error rate for various settings of $\theta$ and signal sizes $N$.  Note that, when $\theta=\pi/2$, the Key states $\ket{\phi_j}$ are actually orthogonal.  As $\theta$ decreases, the maximal tolerated error rate decreases, as shown also asymptotically in \cite{lucamarini2009robust}.}\label{tab:orgedcb091}
\centering
\begin{tabular}{|l|lll|}
\hline
Number of Signals & \(\theta = \pi/2\) & \(\theta = \pi/3\) & \(\theta = \pi/4\)\\[0pt]
\hline
\(N=10^5\) & \(Q=7.6\%\) & \(Q=3.7\%\) & \(Q=1.5\%\)\\[0pt]
\(N=10^8\) & \(Q=10.8\%\) & \(Q=6.8\%\) & \(Q=4.4\%\)\\[0pt]
\hline
\end{tabular}
\end{table}




\section{Closing Remarks}

In this paper, we revisited the so-called Extended B92 QKD protocol, originally introduced in \cite{lucamarini2009robust}, and derived a new, rigorous, proof of security for this protocol in the finite-key scenario.  Our proof did not require any approximation methods to bound the quantum min entropy; we were also able to derive a fairly general result which may hold broader application to other, similar, quantum cryptographic protocols.

Many interesting open problems remain.  We only considered single qubit sources -- analyzing multi-qubit sources and photon loss would be highly beneficial.  We suspect our proof method can easily accommodate photon loss, by extending the underlying alphabet of the sampling strategy to include a third ``vacuum'' state, in addition to the $0$ and $1$ states currently.  One would also need to extend the sampling strategy to ``count'' the number of vacuum events, however this is readily done, as shown in \cite{yao2022quantum} for BB84.  Multi-photon sources, however, would prove more challenging, and we leave that as interesting future work.

Beyond the extended B92 protocol, it would also be interesting to apply our Theorem \ref{thm:result:main} to other QKD protocols (or other quantum cryptographic protocols).  One prime candidate is classical advantage distillation (CAD) \cite{maurer2002secret}, where parties must reject data based on the result of measurements and two-way classical communication.  One should be able to frame this in terms of filtering POVMs and use our theorem, there.  The challenge would be in analyzing the communication leakage, due to the two-way communication, however we feel our method may be suitably adapted to this scenario.

\balance
\bibliographystyle{unsrt}
\bibliography{local}

\begin{thebibliography}{10}

\bibitem{coles2017entropic}
Patrick~J Coles, Mario Berta, Marco Tomamichel, and Stephanie Wehner.
\newblock Entropic uncertainty relations and their applications.
\newblock {\em Reviews of Modern Physics}, 89(1):015002, 2017.

\bibitem{wehner2010entropic}
Stephanie Wehner and Andreas Winter.
\newblock Entropic uncertainty relations—a survey.
\newblock {\em New Journal of Physics}, 12(2):025009, 2010.

\bibitem{pirandola2020advances}
Stefano Pirandola, Ulrik~L Andersen, Leonardo Banchi, Mario Berta, Darius
  Bunandar, Roger Colbeck, Dirk Englund, Tobias Gehring, Cosmo Lupo, Carlo
  Ottaviani, et~al.
\newblock Advances in quantum cryptography.
\newblock {\em Advances in optics and photonics}, 12(4):1012--1236, 2020.

\bibitem{amer2021introduction}
Omar Amer, Vaibhav Garg, and Walter~O Krawec.
\newblock An introduction to practical quantum key distribution.
\newblock {\em IEEE Aerospace and Electronic Systems Magazine}, 36(3):30--55,
  2021.

\bibitem{lucamarini2009robust}
Marco Lucamarini, Giovanni Di~Giuseppe, and Kiyoshi Tamaki.
\newblock Robust unconditionally secure quantum key distribution with two
  nonorthogonal and uninformative states.
\newblock {\em Physical Review A—Atomic, Molecular, and Optical Physics},
  80(3):032327, 2009.

\bibitem{bennett1992quantum}
Charles~H Bennett.
\newblock Quantum cryptography using any two nonorthogonal states.
\newblock {\em Physical review letters}, 68(21):3121, 1992.

\bibitem{duvsek2006quantum}
Miloslav Du{\v{s}}ek, Norbert L{\"u}tkenhaus, and Martin Hendrych.
\newblock Quantum cryptography.
\newblock {\em Progress in optics}, 49:381--454, 2006.

\bibitem{tamaki2009unconditional}
Kiyoshi Tamaki, Norbert L{\"u}tkenhaus, Masato Koashi, and Jamie Batuwantudawe.
\newblock Unconditional security of the bennett 1992 quantum-key-distribution
  scheme with a strong reference pulse.
\newblock {\em Physical Review A—Atomic, Molecular, and Optical Physics},
  80(3):032302, 2009.

\bibitem{wang2026performance}
Zhiyao Wang, Aodh{\'a}n Corrigan, and Norbert L{\"u}tkenhaus.
\newblock Performance of bb84 without decoy states under varying announcement
  structures.
\newblock {\em arXiv preprint arXiv:2603.22448}, 2026.

\bibitem{amer2020finite}
Omar Amer and Walter~O Krawec.
\newblock Finite key analysis of the extended b92 protocol.
\newblock In {\em 2020 IEEE International Symposium on Information Theory
  (ISIT)}, pages 1944--1948. IEEE, 2020.

\bibitem{tomamichel2009fully}
Marco Tomamichel, Roger Colbeck, and Renato Renner.
\newblock A fully quantum asymptotic equipartition property.
\newblock {\em IEEE Transactions on information theory}, 55(12):5840--5847,
  2009.

\bibitem{bouman2010sampling}
Niek~J Bouman and Serge Fehr.
\newblock Sampling in a quantum population, and applications.
\newblock In {\em Annual Cryptology Conference}, pages 724--741. Springer,
  2010.

\bibitem{yao2022quantum}
Keegan Yao, Walter~O Krawec, and Jiadong Zhu.
\newblock Quantum sampling for finite key rates in high dimensional quantum
  cryptography.
\newblock {\em IEEE Transactions on Information Theory}, 68(5):3144--3163,
  2022.

\bibitem{barnett1993eavesdropping}
Stephen~M Barnett, Bruno Huttner, and Simon~JD Phoenix.
\newblock Eavesdropping strategies and rejected-data protocols in quantum
  cryptography.
\newblock {\em Journal of Modern Optics}, 40(12):2501--2513, 1993.

\bibitem{watanabe2008tomography}
Shun Watanabe, Ryutaroh Matsumoto, and Tomohiko Uyematsu.
\newblock Tomography increases key rates of quantum-key-distribution protocols.
\newblock {\em Physical Review A—Atomic, Molecular, and Optical Physics},
  78(4):042316, 2008.

\bibitem{renner2008security}
Renato Renner.
\newblock Security of quantum key distribution.
\newblock {\em International Journal of Quantum Information}, 6(01):1--127,
  2008.

\bibitem{konig2009operational}
Robert Konig, Renato Renner, and Christian Schaffner.
\newblock The operational meaning of min-and max-entropy.
\newblock {\em IEEE Transactions on Information theory}, 55(9):4337--4347,
  2009.

\bibitem{thomas2025new}
Trevor~N Thomas and Walter~O Krawec.
\newblock New key rate bound for high-dimensional bb84 with multiple basis
  measurements.
\newblock {\em To appear: Proc IEEE QCE 2025. arXiv preprint arXiv:2504.11315},
  2025.

\bibitem{huttner1995quantum}
Bruno Huttner, Nobuyuki Imoto, Nicolas Gisin, and Tsafrir Mor.
\newblock Quantum cryptography with coherent states.
\newblock {\em Physical Review A}, 51(3):1863, 1995.

\bibitem{lutkenhaus2002quantum}
Norbert L{\"u}tkenhaus and Mika Jahma.
\newblock Quantum key distribution with realistic states: photon-number
  statistics in the photon-number splitting attack.
\newblock {\em New Journal of Physics}, 4(1):44--44, 2002.

\bibitem{ko2018advanced}
Heasin Ko, Byung-Seok Choi, Joong-Seon Choe, and Chun~Ju Youn.
\newblock Advanced unambiguous state discrimination attack and countermeasure
  strategy in a practical b92 qkd system.
\newblock {\em Quantum Information Processing}, 17(1):17, 2018.

\bibitem{tomamichel2012tight}
Marco Tomamichel, Charles Ci~Wen Lim, Nicolas Gisin, and Renato Renner.
\newblock Tight finite-key analysis for quantum cryptography.
\newblock {\em Nature communications}, 3(1):634, 2012.

\bibitem{maurer2002secret}
Ueli~M Maurer.
\newblock Secret key agreement by public discussion from common information.
\newblock {\em IEEE transactions on information theory}, 39(3):733--742, 2002.

\end{thebibliography}

\end{document}